\documentclass[12pt]{article}
\usepackage{amsmath}

\def\qed{\relax\ifmmode\hskip2em \fbox{ }\else\unskip\nobreak\hskip1em 
$\fbox{}$\fi}

\newsavebox{\theorembox}
\newsavebox{\lemmabox}
\newsavebox{\corollarybox}
\newsavebox{\propositionbox}
\newsavebox{\examplebox}
\newsavebox{\propertybox}
\savebox{\theorembox}{\bf Theorem}
\savebox{\lemmabox}{\bf Lemma}
\savebox{\corollarybox}{\bf Corollary}
\savebox{\propositionbox}{\bf Proposition}
\savebox{\examplebox}{\bf Example}
\savebox{\propertybox}{\bf Property}
\newtheorem{theorem}{\usebox{\theorembox}}
\newtheorem{lemma}[theorem]{\usebox{\lemmabox}}

\newtheorem{definition}{{\sc Definition}\rm }[section]
\newtheorem{definitions}[definition]{{\sc Definitions\rm }}

\newenvironment{proof}{{\noindent\bf Proof~}}{\(\qed\)\vspace*{\proofskip} }
\newlength{\proofskip}
\begin{document}

\begin{center}

{\LARGE 
An Optimal Algorithm for 1-D Cutting Stock Problem
}


\footnotesize

\mbox{\large Srikrishnan Divakaran}\\
School of Engineering and Applied Sciences, Ahmedabad University, 
Ahmedabad, Gujarat, India 380 009,
\mbox{
Srikrishnan.divakaran@ahduni.edu.in }\\[6pt]
\normalsize
\end{center}

\baselineskip 20pt plus .3pt minus .1pt


\noindent 
\begin{abstract}
We present an   $n\Delta^{O(k^2)}$     time algorithm to obtain an optimal solution        for $1$-dimensional cutting stock problem: the bin packing problem of packing $n$ items           onto unit capacity  bins under  the restriction that the number of item sizes $k$ is fixed, where  $\Delta$ is the reciprocal of      the size of the smallest item. We employ elementary ideas in both the design and analysis our algorithm.
\end{abstract}
\noindent\hrulefill
\section{Introduction}
Let      $L$ be an input sequence $L= (a_1,a_2, ...,a_n)$ of $n$ items such  that their sizes are  in the interval $(0, 1)$ and   the number of distinct item sizes is at most $k$. Our goal is to pack the items in $L$ that uses a minimum number   of   unit capacity bins. This problem is a special case of {\em bin packing} where the number of  distinct item sizes is restricted to $k$ and is        often referred to as the {\em 1-dimensional cutting stock problem}. The study of this problem was initiated by     Gilmore and Gomory \cite{GG61} and has  a   wide variety of       applications  \cite{JDUGG74}  including cutting stock applications, packing  problems   in   supply chain management, and resource  allocation problems in distributed   systems. The  bin packing problem  for arbitrary number of item types is        known  to  be strongly NP-Hard \cite{J92} and         hence the interest in the design of efficient algorithms and approximation schemes for bin packing under the  restriction that the number of item types $k$        is fixed.  For  $k=2$, McCormick et al.  \cite{MSS97}    presented     a polynomial time algorithm for determining an optimal solution. For $k \ge 3$, Filippi et al. \cite{FA05} generalized the argument in \cite{MSS97} to    find a polynomial time solution that uses at most $k-2$ bins more  than an optimal       solution. For $d \ge 3$, Jansen and Solis-Oba \cite{JSO10} presented         an algorithm that uses $OPT + 1$ bins and runs in $2^{2^{O(k)}}* (log \Delta)^{O(1)}$, where $\Delta$ is the reciprocal of the size of the smallest item. More  recently, polynomial time   solvability of this problem was settled by the  $(log\Delta)^{2^{O(k )}}$ time optimal algorithm of Goemans et al. \cite{GR14}.   For  arbitrary $k$, the first PTAS was presented by Fernandez de la vega et al \cite{FL81}
, subsequently Karmarkar and Karp \cite{KK82} presented an asymptotic FPTAS that uses    $OPT + O(log^{2}(k))$ bins and Rothvob \cite{R13}  presented an algorithm that          uses $OPT + O(log(k)*log(log(k)))$ bins. Both these algorithms run in pseudo   polynomial time (polynomial on $\sum_{[1..k]}^{} a_i$, where     $a_i$ is the   number of items of the $i$th size). \newline \newline	
{\bf Our Results}:
In  this paper we present an $n\Delta^{O(k^2)}$  time algorithm       to obtain an optimal solution for the bin packing problem under  the restriction that the  number of item sizes $k$ is fixed.
Our main contribution is in the  use of elementary ideas in both the design and analysis our algorithm. The     recent result of    Goemans and Rothvob \cite{GR14} is theoretically a stronger result than ours, but    for moderate to large $k$ our algorithm is computationally better.
\section{An Optimal Algorithm for $1$-dimensional Cutting Stock Problem}
In this section, we first present     some  necessary terms and definitions before presenting an algorithm  that given an input sequence $L$ consisting  of $n$ items from at most $k$ distinct     item sizes determines an optimal bin packing in $O(n\Delta^{k^2})$, where       $\Delta =  \max_{i}^{}\lceil \frac{1} {s_i} \rceil$ be the reciprocal of the smallest item size in $L$. 
\subsection{Preliminaries}
\begin{definitions} 
The sequence $L=(a_1,a_2, ..., a_n)$ with $k$ distinct item sizes \{   $s_1,  s_2,  ..., s_k$\} can be viewed as 
a  $k$ dimensional vector $\vec{d(L)} = (n_1*s_1, n_2 *s_2,...,n_k*s_k)$, where for  $i\in [1..k]$, $n_i$ is the number of items of type  $i$ (size $s_i$); we  refer to $\vec{d(L)}$ as the distribution vector corresponding to $L$. Given   the distribution vector $\vec{d(L)}$, we define  $N^{1}_{L}(\vec{d(L)})$, the unit  neighborhood of $\vec{d(L)}$ with         respect to $L$, to be the collection of points in $R^{k}$ such that each point in this collection     is within unit distance ($l_1$ norm) from some point on $\vec{d(L)}$, and for $i \in [1..k]$, its $i$th  component is an integer multiple of $s_i$.
\end{definitions}
\begin{definitions} 
A unit capacity bin $B$ with a collection of items from $L$ and  a free space  of   at most $\delta$, for   some $\delta\in [0,1]$, can be characterized by a $k$-dimensional vector whose $i$th component, $i\in [1..k]$, equals sum of items of size $s_i$ in $B$. We    refer to   such a vector as $(1-\delta)$ vector consistent with $L$. An $(1-\delta)$ vector that is consistent with $L$      is maximal if addition of any item from $\{s_1, s_2,...,s_k \}$ to     the corresponding bin results   in   the bin exceeding its capacity. 
\end{definitions}
{\bf Note}: We   use   the $l_1$ norm to measure distances and sometimes use vectors and points interchangeably.
\begin{definitions}
Let            $\cal{V}(L)$   denote the  set   of all maximal  $(1-\delta)$ vectors  consistent with $L$,   for $\delta \in (0, 1]$.  For     a given multi-set $V_c(L)=\{\vec{v_1},\vec{v_2},..., \vec{v_c}\}$ of  $c$  vectors   from $\cal{V}(L)$        not necessarily distinct and a point $\vec{p}\in R^{k}$, let $Sum^{p}(V_{c}(L))=\vec{p} +\sum_{i\in [1..c]}^{} \vec{v_i}$ denote the resulting vector obtained by starting at $\vec{p}$ and then placing the $c$ vectors in $V_c(L)$. Given  a   point $\vec{p}\in R^{k}$, 
\begin{itemize}
\item [-] We say that starting at $\vec{p}$ the multi-set   $V_{c}(L)$  crosses  $\vec{d(L)}$ at point $q$ if $q \in \vec{d(L)}$ and   $Sum^{p}(V_{c}(L))$ dominates $q$ (i.e. each of the $k$ components  of $q$ is less than or equal to the corresponding component of $Sum^{p}(V_{c}(L))$;
\item [-] We define $Reachable^{p}(L)$, the  set of points reachable from $\vec{p}$ with respect to $L$, as  the collection of points such that each point is   in $N^{1}_{L}(\vec{d(L)})$ and can be expressed as $Sum^{p}(V_{c} (L))$ for some  $V_{c}(L) \in \        $$\cal{V}$$(L)$, where $c\in [1..k]$;
\item [-] For $q \in Reachable^{p}(L)$, we    define  $RSET^{p}(q)$ to  be     a smallest sized multi-set $V \in \cal{V}(L)$ for which $q=Sum^{p}(V))$. 
\end{itemize}
\end{definitions}
\begin{lemma}
\label{one}
Let $L=(a_1,a_2,..., a_n)$ be a sequence with $k$ distinct item sizes   \{ $s_1,s_2,...,s_k$\}. Let $V^{OPT}(L)$
be the multi-set consisting of $(1-\delta)$-vectors from $\cal{V}(L)$ (i.e.    bin configurations) corresponding 
to an optimal allocation  of items in    $L$ onto unit capacity bins. For some $c \in (1,k+1]$   there  exists a $c$-sized minimal sub-set   $V^{OPT}_{c}(L)\subseteq  V^{OPT}(L)$ such that   starting at $\vec{0}$ or any point that    dominates $\vec{0}$ and is within a unit distance of $\vec{0}$, crosses a point $p$ on $\vec{d(L)}$ at a distance less than or equal to $k$ from the starting point.
\end{lemma}
\begin{proof}
We prove this lemma by contradiction. Let   us    assume that for any   $c \in (1, k+1)$ there does not  exist a multi-set $V^{OPT}_{c} \subseteq V^{OPT}(L)$ of   $c$   vectors   that  crosses $\vec{d(L)}$ at any  point whose distance from the starting point is in the interval $(0, k]$. Then if  we  partition the vectors in $V^{OPT}(L)$ into multi-sets of $c$  vectors, then      since none of the $c$  sized  multi-set   crosses $\vec{d(L)}$ in the interval $(0, k]$ the vector   sum of all the $c$      sized multi-set  will not dominate $\vec{d(L)}$ and hence will    not cross the tip of $\vec{d(L)}$ implying  that $V^{OPT}(L)$ is  missing  some item in $L$ and hence is not a feasible solution. Hence the result.
\end{proof} 
\subsection{Our Algorithm}
Let $V^{OPT}(L)$ be a collection of $(1-\delta)$-vectors from $\cal{V}(L)$ (configuration of unit capacity bins) used for optimally packing items in $L$. From Lemma $\ref{one}$, we  can  observe  that   there  is  a $c$-sized minimal multi-set $V_{c}^{OPT}(L)\subseteq V^{OPT}(L)$, where  $c \in  (0, k+1]$, such     that  starting at the   origin $\vec{0}$ (or any point that dominates $\vec{0}$ and is within unit distance from $\vec{0}$)  the  vector $Sum^{\vec{0}} (V_{c}(OPT(L)))$  crosses  $\vec {d(L)}$ at a point $p$ which is at a $l_1$ distance $d\in (0,k]$ from $\vec{0}$ and ends up at a point within an unit distance of $p$. This implies, we   can  obtain  an optimal solution by (i) determining   the minimum sized  multi-set from $\cal{V}(L)$ that  starts at the origin, crosses $\vec{d(L)}$ 
at $\vec{p}$ whose distance   $d$ from the origin is less than or equal to  $k$ and ends up   at some point that 
is within an unit distance of $\vec{p}$;  and (ii) determine the  minimum sized multi-set from $\cal{V}(L)$ that starts at the point where the first sub-problem ends, crosses    the tip of $\vec{d(L)}$, and ends up at a point within  a unit distance of the tip of $\vec{d(L)}$. Solving these  two sub-problems essentially involves solving the following problem: 
\begin{quote}
Given the start and end points $\vec{v_b}$ and         $\vec{v_e}$ respectively both in $N^{1}_{L}(\vec{d(L)})$, determine a minimum sized collection $V$ of $(1-\delta)$ vectors from $\cal{V}(L)$ such        that $\vec{v_e} = Sum^{\vec{v_b}}(V) \in Reachable^{\vec{b}}(L)$.
\end{quote}
We   now present a recursive algorithm for solving the above problem that can be easily converted into a dynamic 
program.
\begin{tabbing}
{\bf ALGORITHM A($L$, $\vec{v_b}$, $\vec{v_e}$)} \\
	Input(s): \=  (1) \= $L$\ \ \ \ \ \ \ \ \ \ \ \= = $(a_1,a_2,...,a_n)$ \  be the \ sequence \    of \ $n$
	items with  $k$ distinct  item \\ 
	\>      \>        \> \ \ \ \ sizes    \{ $s_1,s_2,...,s_k$\}; \\
	\>            (2) \> $\vec{v_b}$   \>= $(b_1, b_2, ..., b_k)$ be the coordinate of the beginning point \\
	\>            (3) \> $\vec{v_e}$   \>= $(e_1, e_2, ..., e_k)$ be the the coordinate of end point \\
	Output(s): \= The\  collection  of  $(1-\delta)$ \ vectors \ consistent \  with \ $L$ \ for\  some 
	\ $\delta \in (0, 1)$ \  (i.e. \\
	\>                (configuration of unit capacity bins) \ of \ an optimal bin packing of items in $L$; \\
	{\bf Pre}\={\bf -processing}: \\
	\> (1) \= Construct $N^{1}_{L}(\vec{d(L)})$ the unit neighborhood of $\vec{d(L)}$ with respect to $L$;
	\\
    \> (2) \= For \= an integer $c \in [1..k]$\\
	\>      \> Begin \\
	\> (2a)\>     \> Construct $\cal{V}$$_{c}$ the collection containing   all possible $c$-sized multi-sets of
	 $(1-\delta)$\\
	\>     \>  	  \> vectors consistent with $L$; \\
	\> (2b)\>     \> For \= each point $\vec{p} \in N^{1}_{L}(\vec{d(L)})$ and for each $V \in \cal{V}$$_{c}$ \\
	\>     \>     \> Begin \\
	\>     \>     \>     \> If \= $(Sum^{\vec{p}}(V) \in N^{1}_{L}(\vec{d(L)})$ then \\
    \>     \>     \>     \>    \> $Reachable^{\vec{p}}(L)=Reachable^{\vec{p}}(L) \cup Sum^{\vec{p}}(V)$; \\
	\>     \>     \>     \> If \= $(RSET^{\vec{p}}(Sum^{\vec{p}}(V)) == NULL)$ then \\ 
	\>     \>     \>     \>    \> $RSET^{\vec{p}}(Sum^{\vec{p}}(V)) = V$; \\
    \>     \>     \> End \\
    \>     \>  End \\
	{\bf Recursion}: \\
	Beg\=in \\
	\> If \  \ \= $(\vec{v_{e}} \in Reachable^{\vec{v_b}}(L))$ then \\
	\>         \> return $|RSET^{\vec{v_b}}(\vec{v_e})|$; \\
	\> else\= \\
	\>     \>  return $\min_{\vec{p} \in Reachable^{\vec{v_b}}(L)}^{} \{ |RSET^{\vec{v_b}}(\vec{p})| + 
                                                                          A(L, \vec{p},  \vec{v_e}) \}$; \\ 
	End 
\end{tabbing}
{\bf Observation $1$}: The size $|N^{1}_{L}(\vec{d(L)})|$ of the   unit neighborhood of the distribution vector $\vec{d(L)}$ consistent with $L$ is $O(n\Delta^{k})$ and can be     computed in $O(n\Delta^{k})$ time; \newline 
{\bf Observation $2$}: The size $\cal{V}(L)$ of the number of $(1-\delta)$-vectors consistent     with $L$, for 
$\delta \in (0, 1]$, is upper bounded by $\Delta^{k}$, and the size of $\cal{V}$$_{c}(L)$, for $c\in [1..k]$ is
upper bounded by $\Delta^{k} \choose{k}$ = $(\Delta^{k})^{k} = \Delta^{k^2}$; \newline
{\bf Observation $3$}: The number of distinct sub-problems solved by ALGORITHM A($L$, $\vec{v_b}$, $\vec{v_e}$)
is upper bounded by $O(n^2\Delta^{2k})$.
\begin{theorem}
Let $L=(a_1,a_2,...,a_n)$ be a sequence with $k$ distinct item sizes \{ $s_1,s_2,...,s_k$\} and $N^{1}_{l}(\vec
{d(L)})$ be the collection of points in the $1$-neighborhood of $\vec{d(L)}$ consistent with $L$, where $\vec{d
(L)}$  is   the distribution vector of $L$, then the minimum number of bins used to pack the items in $L$ using
unit capacity bins is $\min_{\vec{e}\in N^{1}_{L}(\vec{d(L)})}^{} A(L,\vec{0},\vec{e})$   and can be determined 
in $O(n\Delta^{k^2})$ time.
\end{theorem}
\begin{proof}
Let $V^{OPT}(L)$ be a   collection of $(1-\delta)$-vectors from $\cal{V}(L)$   used for optimally packing items 
in $L$. From   Lemma $\ref{one}$, we  can observe that there is a $c$-sized   minimal multi-set $V_{c}^{OPT}(L) \subseteq V^{OPT}(L)$, where $c \in  (0, k+1]$, such that  starting at the   origin, the  vector $Sum^{\vec{0}} (V_{c}(OPT)(L))$ crosses $\vec {d(L)}$ at a point whose $l_1$ distance $d\in (0,k]$ from the origin and ends up 
at a point within an unit distance from that crossing point. Hence,     we   can obtain an  optimal solution by
(i) determining  the minimum sized  multi-set from $\cal{V}(L)$ that starts at the origin, crosses $\vec{d(L)}$ 
at $\vec{p}$ whose distance  $d$ from the origin is less than or equal to  $k$ and ends up   at some point that 
is within an unit distance of $\vec{p}$;  and (ii) determine the minimum sized multi-set from $\cal{V}(L)$ that starts at the point where the first sub-problem ends, crosses   the tip of $\vec{d(L)}$, and ends up at a point within  a unit distance of the tip of $\vec{d(L)}$. This essentially 
reduces to solving the following problem recursively: 
Given the start and end points $\vec{v_b}$ and         $\vec{v_e}$ respectively both in $N^{1}_{L}(\vec{d(L)})$, determine a minimum sized collection $V$ of $(1-\delta)$ vectors from $\cal{V}(L)$ such        that $\vec{v_e} = Sum^{\vec{v_b}}(V) \in Reachable^{\vec{b}}(L)$. For this  recursive formulation, using a cut and paste argument we can easily 
see that the sub-structure property is true. \newline 
\newline 
Now, we establish the computational complexity of our algorithm. From Observations $1$ and $2$, we can notice that the run-time for the pre-processing steps $1$ and $2$ is dominated by Step $2b$ and is $n\Delta^{O(k^2)}$. From Observation $3$,  we can observe the number of distinct sub-problems is $O(n^2\Delta^{2k})$ and each sub-problem  depends on at most 
$O(n\Delta^{k})$ sub-problems, hence the run-time for the recursion is $O(n^{3}\Delta^{3k})$. Combining     the 
pre-processing time and the time for recursion, we get the run-time of our algorithm is  $max \{ n\Delta^{O(k^2
)}, O(n^{3}\Delta^{3k}) \}$. 
\end{proof}
\section{Conclusions}
We      have designed a psudo-polynomial time algorithm for the Bin Packing problem under the restriction that the number of item types is at most $k$.
For large $k$, our    result    is   stronger than the currently best known theoeretical result of Goemans et al [GR14]. We are       in the process of analyzing        heuristics that are relaxed version of this algorithm that consider    only   a small   fraction of the bin configurations (please see Divakaran \cite{D19}) that in  practice provide good approximations and scale well computationally. 

\end{document}